\documentclass[11pt,a4paper]{article}
\usepackage{amssymb,amsmath,amsthm}
\usepackage[english]{babel}
\usepackage{a4wide}
\usepackage{microtype}
\usepackage{hyperref}
\hypersetup{colorlinks=true,citecolor=blue,linkcolor=blue,urlcolor=blue}
\usepackage{dsfont}
\usepackage{tikz}
\usepackage{float}
\usetikzlibrary{arrows}
\graphicspath{ {./Images/} }

\newtheorem{theorem}{Theorem} 
\newtheorem{lemma}[theorem]{Lemma} 
 
\newtheorem{proposition}[theorem]{Proposition} 
 
\theoremstyle{definition} 
\newtheorem{remark}[theorem]{Remark} 
\newtheorem{example}[theorem]{Example} 
\newtheorem{definition}[theorem]{Definition}

\newcommand{\bin}{\text{bin}}
\newcommand{\zeros}{\text{zeros}}
\newcommand{\vol}{\text{vol}}
\newcommand{\val}{\text{Val}}
\newcommand{\cut}{\textsf{Cut}}

\newcommand{\rep}{\text{rep}}
\newcommand{\zs}{\text{Z}}

\newcommand{\CSPP}{\textnormal{CSP($P$)}}

\begin{document}

\author{Eden Pelleg\\
University of Oxford
\and
Stanislav \v{Z}ivn\'y\\
University of Oxford
}

\title{Additive Sparsification of CSPs\thanks{An extended abstract of this work
appeared in the \emph{Proceedings of the 29th European Symposium on Algorithms} (ESA'21)~\cite{pz21:esa}. Stanislav \v{Z}ivn\'y was supported by a Royal Society University Research Fellowship. This project has received funding from the European Research Council (ERC) under the European Union's Horizon 2020 research and innovation programme (grant agreement No 714532). The paper reflects only the authors' views and not the views of the ERC or the European Commission. The European Union is not liable for any use that may be made of the information contained therein. This work was also supported by UKRI EP/X024431/1. For the purpose of Open Access, the authors have applied a CC BY public copyright licence to any Author Accepted Manuscript version arising from this submission. All data is provided in full in the results section of this paper.}}

\maketitle

\begin{abstract}

  Multiplicative cut sparsifiers, introduced by Bencz\'ur and
  Karger~[STOC'96], have proved extremely influential and found various
  applications. Precise characterisations were established for sparsifiability
  of graphs with other 2-variable predicates on Boolean domains by Filtser and
  Krauthgamer~[SIDMA'17] and non-Boolean domains by Butti and
  \v{Z}ivn\'y~[SIDMA'20].

  Bansal, Svensson and Trevisan [FOCS'19] introduced a weaker notion of
  sparsification termed ``additive sparsification'', which does not require
  weights on the edges of the graph. In particular, Bansal et al. designed
  algorithms for additive sparsifiers for cuts in graphs and hypergraphs.
   
  As our main result, we establish that \emph{all} Boolean Constraint
  Satisfaction Problems (CSPs) admit an additive sparsifier; that is, for every
  Boolean predicate $P:\{0,1\}^k\to\{0,1\}$ of a fixed arity $k$, we show that
  $\CSPP$ admits an additive sparsifier. 
  Under our newly introduced notion of all-but-one sparsification for
  non-Boolean predicates, we show that $\CSPP$ admits an additive sparsifier for
  \emph{any} predicate $P:D^k\to\{0,1\}$ of a fixed arity $k$ on an arbitrary
  finite domain $D$.

\end{abstract}

\section{Introduction}
\label{sec:intro}

Graph sparsification is the problem of, given a graph $G = (V,E)$ with
quadratically many (in $|V|$) edges, 
finding a sparse subgraph $G_\varepsilon =
(V,E_\varepsilon\subseteq E)$ such that important properties of $G$ are
preserved in $G_\varepsilon$. Sparse in this context usually means with sub-quadratically
many edges, though in this work we require (and can achieve) linearly many edges.

One of the most studied properties of preservation is the size of cuts. If $G =
(V,E,w)$ is an undirected weighted graph with $w:E\to \mathbb{R}_{>0}$, given
some $S\subseteq V$, the cut of $S$ in $G$ is 
\[\cut_G(S) = \sum_{\substack{\{u,v\}\in E\\|\{u,v\}\cap S| = 1}}w(\{u,v\}),\]
the sum of
weights of all edges connecting $S$ and $S^c = V\setminus S$. In an influential
paper, Bencz\'{u}r and Karger~\cite{mult_nlogn} introduced cut sparsification
with a multiplicative error. In particular,~\cite{mult_nlogn} showed that for
any
graph $G = (V,E,w)$ and any error parameter $0<\varepsilon<1$, there exists a
sparse subgraph $G_\varepsilon = (V,E_\varepsilon\subseteq E, w')$ with $O(n(\log
n)\varepsilon^{-2})$ edges (and new weights $w'$ on the edges in
$E_\varepsilon$), such that for every
$S\subseteq V$ we have 
\[\cut_{G_\varepsilon}(S)\in (1\pm \varepsilon)\cut_G(S).\]
This was later improved by Batson, Spielman and
Srivastava~\cite{mult_basic} to a subgraph with $O(n\varepsilon^{-2})$ many
edges. Andoni, Chen, Krauthgamer, Qin, Woodruff and Zhang showed that the
dependency on $\varepsilon$ is optimal~\cite{mult_optimality}.

The ideas from cut sparsification paved the way to various generalisations,
including streaming~\cite{Ahn09:icalp}, sketching~\cite{mult_optimality}, cuts
in hypergraphs~\cite{cut_hypergraphs,Newman13:sicomp}, spectral
sparsification~\cite{Spielman04:stoc,Spielman11:sicomp-graph,cut_spectral,Fung19:sicomp,spectral_hypergraphs}
and the consideration of other predicates besides
cuts~\cite{mult_other_predicates}. In this work, we focus on the latter.

The cut sparsification result in~\cite{mult_basic} was explored for other
Boolean binary predicates by Filtser and
Krauthgamer~\cite{mult_other_predicates}, following a suggestion to do so by
Kogan and Krauthgamer in~\cite{cut_hypergraphs}. Filtser and Krauthgamer
found~\cite{mult_other_predicates} a necessary and sufficient
condition on the predicate for the graph to be sparsifiable (in the sense
of~\cite{mult_basic}). In particular, \cite{mult_other_predicates} showed that not all
Boolean binary predicates are sparsifiable. Later, Butti and
\v{Z}ivn\'{y}~\cite{mult_larger_domains} generalised the result
from~\cite{mult_other_predicates} to arbitrary finite domain binary
predicates.

We remark that~\cite{mult_other_predicates,mult_larger_domains} use the
terminology of \emph{constraint satisfaction problems} (CSPs) with a fixed
predicate $P$. This is is equivalent to a (hyper)graph $G$ with a fixed
predicate. Indeed, the vertices of $G$ correspond to the variables of the CSP
and the (hyper)edges of $G$ correspond to the constraints of the CSP. If the
fixed predicate $P$ is not symmetric, the (hyper)edges of $G$ are directed. We
will mostly talk about sparsification of (hyper)graphs with a fixed predicate
but this is equivalent to the CSP view.

Recently, while trying to eliminate the requirement for the introduction of new
weights for the sparse subgraph, Bansal, Svensson and Trevisan~\cite{additive}
have come up with a new sparsification notion with an additive error term. They
have shown (cf. Theorem~\ref{thm:additive_cut} in Section~\ref{sec:prelims})
that under their notion any undirected unweighted hypergraph has a sparse
subhypergraph which preserves all cuts up to some additive term. 

\paragraph{Motivation}
The relatively recent notion of additive sparsification has not yet been
explored to the same extent as the notion of multiplicative sparsification has
been. We believe that this notion has a lot of potential for applications as the
sparsifiers are not weighted, unlike multiplicative sparsifiers, and the main
restriction of multiplicative sparsifiers in applications appears to be the
number of distinct weights required in sparsifiers. For some graphs (such as the
``barbell graph'' -- two disjoint cliques joined by a single edge), any
nontrivial multiplicative sparsifier requires edges of different weights. In any
case, the authors find the notion of additive sparsification interesting in its
own right, independently of applications. We refer the reader to~\cite{additive}
for further details and a discussion.

The goal of our work is to understand how the notion of additive sparsification
developed in~\cite{additive} for cuts behaves on (hyper)graphs with other
predicates (beyond cuts), deriving inspiration from the generalisations of cuts
to other predicates in the multiplicative setting established
in~\cite{mult_other_predicates, mult_larger_domains}. In particular, already
Boolean binary predicates include interesting predicates such as the \emph{uncut
edges} (using the predicate $P(x,y)=1$ iff $x=y$), \emph{covered edges} (using
the predicate $P(x,y)=1$ iff $x=1$ or $y=1$), or \emph{directed cut edges} (using
the predicate $P(x,y)=1$ iff $x=0$ and $y=1$). While such graph problems are
well-known and extensively studied, it is not clear whether one should expect
them to be sparsifiable or not. For instance, as mentioned before, not all (even
Boolean binary) predicates are sparsifiable
multiplicatively~\cite{mult_other_predicates}. Are there some predicates that
are not additively sparsifiable?

\subsection{Contributions}

\paragraph{Boolean predicates}
Our main result, Theorem~\ref{thm:additive_boolean} in Section~\ref{sec:main}, shows that
\emph{all} hypergraphs with \emph{constant} uniformity $k$, directed or undirected,
admit additive sparsification with respect to all Boolean predicates $P:\{0,1\}^k\to
\{0,1\}$; the number of hyperedges of the sparsifier with error
$\varepsilon>0$ is $O\left(n\varepsilon^{-2}\log\frac{1}{\varepsilon}\right)$,
where the $O(\cdot)$ hides a factor that depends on $k$.
This result has three ingredients. First, we
observe that the result in~\cite{additive} also holds true for directed
hypergraphs. Second, we use a reduction via the $k$-partite $k$-fold covers of
hypergraphs to the already solved case of Boolean \cut. Finally, we use linear
algebra to prove the correctness of the reduction. While the reduction via the
$k$-partite $k$-fold cover was used in previous works on multiplicative
sparsification~\cite{mult_other_predicates,mult_larger_domains},
the subsequent non-trivial linear-algebraic
analysis (Proposition~\ref{prp:even_not_all_one}) is novel and constitutes our
main technical contribution, as well as our result that, unlike in the
multiplicative setting, all Boolean predicates can be (additively) 
sparsified. 
We also show that our results immediately apply to the more general setting
where different hyperedges are associated with different predicates (cf.
Remark~\ref{rem:gen}). This corresponds to CSPs with a fixed constraint language
(of a finite size) rather than just a single predicate.

\paragraph{Non-Boolean predicates}
We introduce a notion of sparsification that generalises the Boolean case to
predicates on non-Boolean domains, i.e.\ a notion capturing predicates of the
form $P:D^k\to\{0,1\}$, where $D$ is an arbitrary fixed finite set with $|D|\geq
2$. We call this type of sparsification ``all-but-one'' sparsification since
the additive error term includes the maximum volume of $|D|-1$ (out of $|D|$) parts,
where the volume of a subset is the sum of the degrees in the subset. (The
precise definition can be found in Section~\ref{sec:non-Boolean}.)
By building on the techniques used to establish our main result, we show that all
hypergraphs (again, directed or undirected) admit additive all-but-one
sparsification with respect to all predicates. This is stated as Theorem~\ref{thm:larger_domains}
in Section~\ref{sec:non-Boolean}. 
We also show, in Section~\ref{sec:optimal}, that our notion of all-but-one
sparsification is, in some sense, optimal. 

\paragraph{Comparison to previous work}
As mentioned above,
our sparsifiability result is obtained by a reduction via the $k$-partite
$k$-fold cover to the cut case established in~\cite{additive}. A reduction via
the $k$-partite $k$-fold cover was also used (for $k=2$) in previous work on multiplicative sparsification~\cite{mult_other_predicates,mult_larger_domains}. In particular, the
correctness of the reduction for Boolean
binary predicates in~\cite{mult_other_predicates} is done via an ad hoc case
analysis for 11 concrete predicates. In the generalisation to binary predicates
on arbitrary finite domains in~\cite{mult_larger_domains}, the correctness is 
proved via a combinatorial property of bipartite graphs without a certain
$4$-vertex graph\footnote{A bipartite graph on four vertices with
each part of size two and precisely one edge between the two parts.}
as a subgraph and a reduction to cuts with more than two parts. 

In our case, we use the same black-box reduction via the $k$-partite $k$-fold
cover. Thus the reduction itself is pretty straightforward, although the
analysis is not. In fact, we find it surprising and unexpected that the
$k$-partite $k$-fold cover works in the additive setting. Our key contribution
is the proof of its correctness. A few simple reductions get us to the most
technically involved case, in which $k$ is even and the $k$-ary predicate
satisfies $P(1,\ldots,1)=0$. Additive sparsifiability of such predicates is
established in Proposition~\ref{prp:even_not_all_one}. Unlike in the
multiplicative setting, it is not clear how to do this in a straightforward way
similar to~\cite{mult_other_predicates,mult_larger_domains}. Instead, we
associate with a given predicate $P$ a vector $v_P$ in an appropriate vector
space, identify special vectors that can be shown additively sparsifiable
directly, show that linear combinations preserve sparsifiability, and argue that
$v_P$ can be generated by the special vectors. The latter is the most technical
part of the proof. While there are several natural ideas how to achieve this in
a seemingly simpler way (such as arguing that the special vectors form a basis),
we have not managed to produce a simpler or shorter proof.

The result in~\cite{additive} also works for non-constant $k$. We emphasise that
we deal with constant $k$, which is standard in the CSP literature in that
the predicate (or a set of predicates) is fixed and not part of the input. For
constant $k$, the representation of predicates is irrelevant (cf.
Remark~\ref{rmk:not_constant}). Thus we do not keep track of (and have not tried
to optimise) the precise dependency of the reduction on the predicate arity $k$
(or the domain size $q=|D|$).

\paragraph{Related work}
The already mentioned spectral sparsification~\cite{Spielman04:stoc} is a
stronger notion than cut sparsification as it requires that not only cuts but
also the Laplacian spectrum of a given graph should be (approximately)
preserved~\cite{Spielman11:sicomp-graph,cut_spectral,Fung19:sicomp,spectral_hypergraphs,additive}.

Our focus in this article is on \emph{edge sparsifiers} (of cuts and generalisations via
local predicates). There are also vertex sparsifiers, in which one reduces the
number of vertices. Vertex sparsifiers have been studied for cut sparsification
(between special vertices called
terminals)~\cite{Hagerup98:jcss,Moitra09:focs,Leighton10:stoc,Chalermsook21:soda}
as well as for spectral sparsification~\cite{Kyng16:stoc}. 

Sparsification in general is about finding a sparse sub(hyper)graph while preserving
important properties of interest. In addition to cut sparsifiers, another well
studied concept is that of \emph{spanners}.
A spanner of a graph is a (sparse) subgraph that approximately preserves 
distances of shortest paths. Spanners have been studied in great detail both in the
multiplicative~\cite{Awerbuch85:jacm,Peleg89:jgt,Althofer1993:dcg,Cohen98:sicomp,Awerbuch98:sicomp,Baswana07:rsa,Roditty05:icalp}
and
additive~\cite{Aingworth99:sicomp,Dor00:sicomp,Bollobas05:sidma,Baswana05:soda,Woodruff10:icalp,Chechik13:soda}
setting. Emulators are a generalisation of spanners in which the sparse
graph is not required to be a subgraph of the original graph. We refer the reader to a nice
recent survey of Elkin and Neimain for more details~\cite{Elkin20:survey}.

\section{Preliminaries}
\label{sec:prelims}

For an integer $k$, we denote by $[k]$ the set $\{0,1,\hdots,k-1\}$. 
All graphs and hypergraphs\footnote{We use the standard definition of hypergraphs, in
which every hyperedge is an ordered tuple of vertices.} in this paper are unweighted.

For an assignment $a:V\to S$ from the set of vertices of a (hyper)graph to some
set $S$ containing $0$, we denote by $\zs_a = \{v\in V: a(v)=0\}$ the set of
vertices mapped to $0$.

If $0\leq i \leq r^k - 1$ is an integer, we denote by $\rep_{r,k}(i)$ the
representation of $i$ in base $r$ as a vector in $\mathbb{R}^k$, where the first
coordinate stands for the most significant digit, and the last coordinate for
the least significant digit. For the special case $r=2$, we use the notation $\bin_k(i)$
for the binary representation of $i$.

We denote by $v[j]$ the $j$-th coordinate of the vector $v$, counting from $0$. 

For an integer $0\leq i \leq 2^k - 1$, we use $\zeros_k(i) = \{\ell\in[k]:
\bin_k(i)[\ell] = 0\}$; for example $\zeros_6(52) = \{2,4,5\}$, since
$\bin_6(52) = (1,1,0,1,0,0)$.

We now define the value of an assignment on a hypergraph with a fixed predicate.

\begin{definition}
\label{def:values}
  Let $G = (V,E)$ be a directed $k$-uniform hypergraph and let $P:
  D^k\to \{0,1\}$ be a $k$-ary predicate on a finite set $D$. Given an
  assignment $a:V\to D$ of $G$, the \emph{value} of $a$ is defined by
  $\val_{G,P}(a) = \sum_{(v_1,\hdots ,v_k)\in E} P(a(v_1),\hdots ,a(v_k))$.
  If $G$ is undirected and $P$ is order invariant,\footnote{$P(b_1,\hdots,b_k) = P(b_{\sigma(1)},\hdots,b_{\sigma(k)})$ for all $b_1,\hdots,b_k\in D$ and every permutation $\sigma$ on the set $\{1,\ldots,k\}$.} we define
  $\val_{G,P}(a) = \sum_{\{v_1,\hdots ,v_k\}\in E} P(a(v_1),\hdots ,a(v_k))$.\footnote{The terms are well defined since $P$ is order invariant.}
\end{definition}

The notion of additive sparsification was first introduced in~\cite{additive}
for cuts in graphs and hypergraphs. In order to define it, we will need the
$\cut: \{0,1\}^k\to \{0,1\}$ predicate defined by
$\cut(b_1,\hdots ,b_k) = 1 \iff \exists i,j, b_i\ne b_j$.
Given a hypergraph $G=(V,E)$ and a set $U\subseteq V$, we denote by $\vol_G(U)$
the \emph{volume} of $U$, defined as the sum of the degrees in $G$ of all
vertices in $U$.

\begin{definition}
\label{def:additive_cut}
Let $G = (V,E)$ be an undirected $k$-uniform hypergraph, and denote $|V| = n$. We
say that $G$ \emph{admits additive cut sparsification} with error
$\varepsilon$ using $O(f(n,\varepsilon))$ hyperedges if there exists a
subhypergraph $G_\varepsilon = (V, E_\varepsilon\subseteq E)$ with
$|E_\varepsilon| = O(f(n,\varepsilon))$, called an \emph{additive sparsifier}
of $G$, such that for every assignment $a:V\to \{0,1\}$ we have
\begin{equation}
\label{eq:cut_additive}
\left|\frac{|E|}{|E_\varepsilon|}\val_{G_\varepsilon,\cut}(a) - \val_{G,\cut}(a)\right|\leq \varepsilon(d_G|\zs_a| + \vol_G(\zs_a)),
\end{equation}
where $d_G$ is the average degree of $G$.
\end{definition}

Note that (\ref{eq:cut_additive}) can also be written as
\[\frac{|E|}{|E_\varepsilon|}\val_{G_\varepsilon,\cut}(a)\in \val_{G,\cut}(a)
\pm \varepsilon(d_G|\zs_a| + \vol_G(\zs_a)),\] which explains the use of the
term ``additive'' for the error. 

Bansal, Svensson and Trevisan~\cite{additive} showed the following sparsification result: 

\begin{theorem}[\protect{Additive Cut Sparsification~\cite[Theorem 1.3]{additive}}] 
\label{thm:additive_cut}
Let $G = (V,E)$ be an undirected $n$-vertex $k$-uniform hypergraph, and $\varepsilon>0$. Then $G$ admits additive cut sparsification with error $\varepsilon$ using $O\left(\frac{n}{k}\varepsilon^{-2}\log(\frac{k}{\varepsilon})\right)$ hyperedges.
\end{theorem}

\begin{remark}
\label{rmk:undirected}
We call a predicate $P$ \emph{symmetric} if it is order invariant (as in
  Definition~\ref{def:values}). Since Theorem~\ref{thm:additive_cut} deals with
  only \emph{undirected} hypergraphs, it is not clear how to generalise it to non-symmetric predicates directly, since the value of such predicates on undirected hypergraphs is not defined. Therefore, our course of action will be first to prove it for the case of directed hypergraphs, and then generalise it to other predicates on directed hypergraphs. In fact, by doing this we also prove the result for undirected hypergraphs with symmetric predicates, since hyperedges can be given arbitrary directions without changing 
the average degree of $G$, or the volume in $G$, or the value of the predicate in
any assignment.
\end{remark}

\begin{remark}
Throughout this paper we only discuss the existence of sparsifiers and do not
  mention the time complexity to find them. However, the (implicit) time complexity results
  from~\cite{additive} apply in our more general setting as well since the sparsifiers we find are in fact the same sparsifiers for all predicates, including cuts (cf. Remark~\ref{rmk:same_subhypergraph}).
\end{remark}

An important tool we use to prove our results is the $k$-partite
$k$-fold cover of a hypergraph. This construction is a well known one, and has
been used for multiplicative sparsification (for $k=2$)
in~\cite{mult_other_predicates} and~\cite{mult_larger_domains}.

\begin{definition}
\label{def:cover}
Let $G = (V,E)$ be a directed $k$-uniform hypergraph. The \emph{$k$-partite $k$-fold cover} of $G$ is the hypergraph $\gamma(G) = (V^\gamma,E^\gamma)$ where
  \[V^\gamma = \{v^{(0)}, v^{(1)},\hdots ,v^{(k-1)}: v\in V\},\]
  \[E^\gamma = \{(v_1^{(0)}, v_2^{(1)},\hdots ,v_k^{(k-1)}): (v_1,\hdots ,v_k)\in E\}.\]
If $G$ is undirected we define the cover in the same way except 
  \[E^\gamma = \{\{v_{\sigma(1)}^{(0)}, v_{\sigma(2)}^{(1)},\hdots ,v_{\sigma(k)}^{(k-1)}\}: \{v_1,\hdots
  ,v_k\}\in E\}, \mbox{ $\sigma$ a permutation on $\{1,2,\ldots,k\}$}\}\] so for each hyperedge in $G$ we get $k!$ hyperedges in
  $\gamma(G)$ in this case.
\end{definition}
If $k=2$ then $\gamma(G)$ corresponds to the well-known \emph{bipartite double
cover} of $G$~\cite{Brualdi80:jgt}.

\section{Sparsification of Boolean Predicates}
\label{sec:main}

As mentioned in Section~\ref{sec:intro}, we begin by observing that
Theorem~\ref{thm:additive_cut} also works for directed hypergraphs. 
(We emphasise that we treat $k$ as a constant, cf. Remark~\ref{rmk:not_constant}.)

We will need a notation for the undirected equivalent of a directed hypergraph.

\begin{definition}
Given a directed $k$-uniform hypergraph $G = (V,E)$, the \emph{undirected equivalent} of $G$ is $\Lambda(G) = (V, \overline{E})$ where $\overline{E} =
  \{\{v_1,\hdots,v_k\}: (v_1,\hdots,v_k)\in E\}$.
\end{definition}

In other words, $\Lambda(G)$ is obtained by ``forgetting'' the directions of the hyperedges of $G$ (and ignoring duplicates if they exist).

\begin{proposition}
\label{prp:directed_cut}
Let $G = (V,E)$ be a directed $n$-vertex $k$-uniform hypergraph, and $\varepsilon>0$. Then $G$ admits additive cut sparsification with error $\varepsilon$ using $O\left(n\varepsilon^{-2}\log\frac{1}{\varepsilon}\right)$ hyperedges.
\end{proposition}

\begin{proof}
Let $\varepsilon>0$, and let $\gamma(G) = (V^\gamma,E^\gamma)$ be the $k$-partite $k$-fold cover of $G$. Let $\Lambda(\gamma(G))$ be the undirected equivalent of $\gamma(G)$. Let $\Lambda(\gamma(G))_\varepsilon = (V^\gamma_\varepsilon,\overline{E^\gamma_\varepsilon})$ be a subhypergraph of $\Lambda(\gamma(G))$ promised by Theorem~\ref{thm:additive_cut}. By the construction of the $k$-partite $k$-fold cover, there are no two directed hyperedges over the same set of vertices, and so there is a 1-1 correspondence between the hyperedges of $G$ and the hyperedges of $\Lambda(\gamma(G))$. Hence we have a subhypergraph $G_\varepsilon = (V,E_\varepsilon)$ of $G$ such that $\Lambda(\gamma(G_\varepsilon)) = \Lambda(\gamma(G))_\varepsilon$ (by taking the hyperedges corresponding to the ones of $\Lambda(\gamma(G))_\varepsilon$). We also have $|\overline{E^\gamma}| = |E|$ and $|\overline{E^\gamma_\varepsilon}| = |E_\varepsilon|$.

Let $a:V\to\{0,1\}$. Define $a':V^\gamma\to\{0,1\}$ by $a'(v^{(i)}) = a(v)$. We have
\begin{equation}
\label{eq:directed_g}
\val_{G,\cut}(a) = \val_{\Lambda(\gamma(G)),\cut}(a'),
\end{equation}
which is true for any hypergraph, and in particular for $G_\varepsilon$:
\begin{equation}
\label{eq:directed_ge}
\val_{G_\varepsilon,\cut}(a) = \val_{\Lambda(\gamma(G_\varepsilon)),\cut}(a').
\end{equation}

Applying Theorem~\ref{thm:additive_cut} to $\Lambda(\gamma(G))$ and $a'$ gives us
\begin{align*}
\left|\frac{|E|}{|E_\varepsilon|}\val_{G_\varepsilon,\cut}(a) - \val_{G,\cut}(a)\right| & = \left|\frac{|\overline{E^\gamma}|}{|\overline{E^\gamma_\varepsilon}|}\val_{\Lambda(\gamma(G_\varepsilon)),\cut}(a') - \val_{\Lambda(\gamma(G)),\cut}(a')\right| \\
& \leq \varepsilon (d_{\Lambda(\gamma(G))}\left|\zs_{a'}\right|+\vol_{\Lambda(\gamma(G))}(\zs_{a'})) \\ & = \varepsilon(d_{\gamma(G)}\cdot k|\zs_a|+\vol_{\gamma(G)}(\zs_{a'})) \\
& = \varepsilon(d_G|\zs_a|+\vol_G(\zs_a)),
\end{align*}
where the first line is due to (\ref{eq:directed_g}) and (\ref{eq:directed_ge}), the second line is by Theorem~\ref{thm:additive_cut}, and the last two lines are by properties of the $k$-partite $k$-fold cover. Moreover
  \[|E_\varepsilon| = |\overline{E^\gamma_\varepsilon}| =
  O\left(\frac{kn}{k}\varepsilon^{-2}\log\frac{k}{\varepsilon}\right)=O\left(n\varepsilon^{-2}\log\frac{1}{\varepsilon}\right),\]
as required.
\end{proof}

From now on, whenever we say a ``hypergraph'', we mean a ``directed hypergraph''
with $n$ vertices. By Remark~\ref{rmk:undirected}, the results also apply to
undirected hypergraphs (whenever it makes sense, i.e.\ if the associated predicate
is symmetric). We also omit the word additive when discussing sparsification.
The following notion of sparsification is a natural generalisation of cut
sparsification (Definition~\ref{def:additive_cut}) to arbitrary predicates.

\begin{definition}
\label{def:additive_hypergraphs}
Let $P$ be a $k$-ary Boolean predicate and $G = (V,E)$ a $k$-uniform hypergraph.
  We say that $G$ \emph{admits $P$-sparsification} with error $\varepsilon$
  using $O(f(n,\varepsilon))$ hyperedges if there exists a subhypergraph
  $G_\varepsilon = (V, E_\varepsilon\subseteq E)$ with $|E_\varepsilon| =
  O(f(n,\varepsilon))$, called a \emph{$P$-sparsifier} of $G$, such that for every assignment $a:V\to \{0,1\}$ we have
\begin{equation}
\label{eq:main_sparsify}
\left|\frac{|E|}{|E_\varepsilon|}\val_{G_\varepsilon,P}(a) - \val_{G,P}(a)\right|\leq \varepsilon(d_G|\zs_a| + \vol_G(\zs_a)),
\end{equation}
where $d_G$ is the average degree of $G$.
\end{definition}

The following theorem is our main result, extending Proposition~\ref{prp:directed_cut} to \emph{all} $k$-ary predicates with Boolean domains.

\begin{theorem}[\textbf{Main}]
\label{thm:additive_boolean}
For every $k$-uniform hypergraph $G$ ($k$ is a constant), every $k$-ary Boolean
predicate $P:\{0,1\}^k\to\{0,1\}$, and every $\varepsilon>0$, $G$ admits $P$-sparsification with error $\varepsilon$ using $O\left(n\varepsilon^{-2}\log\frac{1}{\varepsilon}\right)$ hyperedges.
\end{theorem}

Theorem~\ref{thm:additive_boolean} can be informally restated as ``\emph{every $k$-uniform hypergraph is sparsifiable with respect to all $k$-ary Boolean predicates}'' or ``\emph{for every Boolean predicate $P$ of constant arity, $\CSPP$ is sparsifiable}''.

\begin{remark}
  \label{rem:gen}
It is possible to consider an even more general case where each hyperedge in 
$G$ has its own predicate. In this case, we can apply
  Theorem~\ref{thm:additive_boolean} to each of the hypergraphs obtained by
  taking only hyperedges corresponding to a specific predicate, and so get a
  sparsifier for each such predicate. Taking the union of all their hyperedges,
  we get a new hypergraph $G_\varepsilon$, which is a sparsifier of the original
  hypergraph. Indeed, it has
  $O\left(n\varepsilon^{-2}\log\frac{1}{\varepsilon}\right)$ hyperedges since
  it is the union of a constant number of hypergraphs. (The number of
  predicates $P:\{0,1\}^k\to\{0,1\}$ is constant, since $k$ is constant.) It
  also satisfies (\ref{eq:main_sparsify}) for any given assignment up to some
  constant factor, since all the sparsifiers it is composed of do. This constant
  factor can be eliminated by choosing $\varepsilon_0 = \frac{\varepsilon}{m}$
  for an appropriate $m$ that depends only on $k$.
\end{remark}

The main work in the proof of Theorem~\ref{thm:additive_boolean} is for even
values of $k$; a simple reduction (Proposition~\ref{prp:all_odd_hypergraphs}) then reduces the case of
$k$ odd to the even case.

In order to prove Theorem~\ref{thm:additive_boolean} for even $k$, we use the
$k$-partite $k$-fold cover of $G$ and apply Proposition~\ref{prp:directed_cut}
to various assignments of it. For a $k$-ary Boolean predicate $P:\{0,1\}^k\to
\{0,1\}$, we consider the vector $v_P\in \mathbb{R}^{2^k}$, defined by $v_P[i] =
P(\bin_k(i))$.
For instance, for the $\cut$ predicate on a 3-uniform hypergraph, we have $v_\cut = (0,1,1,1,1,1,1,0)$.

For a given hypergraph $G$ and an assignment $a$, we consider the vector
$v_{G,a}\in \mathbb{R}^{2^k}$ defined by $v_{G,a}[i] =
\left|\{(v_1,\hdots,v_k)\in E: (a(v_1),\hdots,a(v_k)) = \bin_k(i)\}\right|$.
In other words, each coordinate of $v_{G,a}$ counts the hyperedges in $G$ whose vertices are assigned some specific set of values by $a$. 

\begin{example}
\label{example}
  Given the graph $G=(V,E)$ in Figure~\ref{fig:simple_graph} (so $k=2$) and the assignment
  $a:V\to \{0,1\}$ defined as $a(v_1)=a(v_2)=a(v_3)=0$ and
  $a(v_4)=a(v_5)=a(v_6)=a(v_7)=1$, we have $v_{G,a} = (2,3,1,5)$, since there are two edges with assignment $(0,0)$, namely $(v_1,v_2)$ and
  $(v_2,v_3)$, three edges with assignment $(0,1)$, namely $(v_1,v_4)$,
  $(v_2,v_6)$, and $(v_2,v_7)$, etc.
\end{example}

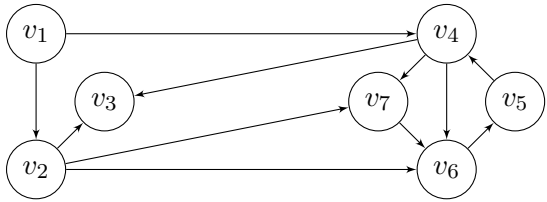
\begin{figure}[H]
\centering
\begin{tikzpicture}
  [scale=.9,auto=center,every node/.style={draw, circle}] 
  \tikzset{vertex/.style = {shape=circle,draw,minimum size=1.5em}}
  \tikzset{edge/.style = {->,> = latex'}}
    
  \node[vertex] (a1) at (0,2) {$v_1$};  
  \node[vertex] (a2) at (0,0) {$v_2$};  
  \node[vertex] (a3) at (1,1)  {$v_3$};
  \node[vertex] (a4) at (6,2) {$v_4$};  
  \node[vertex] (a5) at (7,1)  {$v_5$};
  \node[vertex] (a6) at (6,0) {$v_6$};  
  \node[vertex] (a7) at (5,1)  {$v_7$};
  
  \draw[edge] (a1) to (a2);
  \draw[edge] (a2) to (a3);  
  \draw[edge] (a1) to (a4);
  \draw[edge] (a4) to (a6);
  \draw[edge] (a4) to (a7);
  \draw[edge] (a5) to (a4);
  \draw[edge] (a4) to (a3);
  \draw[edge] (a2) to (a6);
  \draw[edge] (a2) to (a7);
  \draw[edge] (a6) to (a5);
  \draw[edge] (a7) to (a6);
\end{tikzpicture}
\caption{Graph from Example~\ref{example}.}
\label{fig:simple_graph}
\end{figure}

Under these notations, we get $\val_{G,P}(a) = \langle v_P, v_{G,a}\rangle$,
where $\langle\cdot,\cdot\rangle$ is the standard inner product in $\mathbb{R}^{2^k}$.
We begin by proving the following useful lemma.

\begin{lemma}
\label{lem:sparse_by_linear_combination}
Let $G = (V,E)$ be a $k$-uniform hypergraph, $P_1,\hdots,P_m$ be $k$-ary Boolean predicates ($m$ is a constant). Suppose that for every $\varepsilon>0$ and $1\leq i\leq m$, $G$ admits $P_i$-sparsification with error $\varepsilon$ using $O\left(n\varepsilon^{-2}\log\frac{1}{\varepsilon}\right)$ hyperedges, and that the same subhypergraph $G_\varepsilon = (V,E_\varepsilon\subseteq E)$ is a $P_i$-sparsifier for all $P_i$. Suppose that $P$ is some $k$-ary Boolean predicate for which we have
$v_P = \sum_{i=1}^m\lambda_iv_{P_i}$ for some constants $\lambda_1,\hdots,\lambda_m\in\mathbb{R}$.
Under these conditions, $G$ admits $P$-sparsification with error $\varepsilon$ using $O\left(n\varepsilon^{-2}\log\frac{1}{\varepsilon}\right)$ hyperedges.
\end{lemma}

\begin{proof}
Let $\varepsilon>0$ and denote $\varepsilon_i = \frac{\varepsilon}{m|\lambda_i|}$ (if $\lambda_i=0$ take $\varepsilon_i=1$ instead) and $\varepsilon_0 = \min\{\varepsilon_1,\hdots,\varepsilon_m\}$. Let $G_{\varepsilon_0} = (V,E_{\varepsilon_0})$ be the common witness subhypergraph for $\varepsilon_0$ promised by the assumption. We know that every $P_i$ satisfies
\begin{equation}
\label{eq:each_pi}
\left|\frac{|E|}{|E_{\varepsilon_0}|}\val_{G_{\varepsilon_0},P_i}(a) - \val_{G,P_i}(a)\right|\leq \varepsilon_0(d_G|\zs_a| + \vol_G(\zs_a))
\end{equation}
for every assignment $a:V\to \{0,1\}$.
We also have 
\[\val_{G,P}(a) = \langle v_P, v_{G,a}\rangle = \left\langle
  \sum_{i=1}^m\lambda_iv_{P_i}, v_{G,a}\right\rangle =
  \sum_{i=1}^m\lambda_i\langle v_{P_i}, v_{G,a}\rangle = \sum_{i=1}^m
  \lambda_i\val_{G,P_i}(a),\]
and similarly
  \[\val_{G_{\varepsilon_0},P}(a) = \sum_{i=1}^m
  \lambda_i\val_{G_{\varepsilon_0},P_i}(a).\]
Therefore, for every assignment $a$ we get
\begin{align*}
\left|\frac{|E|}{|E_{\varepsilon_0}|}\val_{G_{\varepsilon_0},P}(a) - \val_{G,P}(a)\right| & =  \left|\frac{|E|}{|E_{\varepsilon_0}|}\sum_{i=1}^m\lambda_i\val_{G_{\varepsilon_0},P_i}(a) - \sum_{i=1}^m\lambda_i\val_{G,P_i}(a)\right| \\
& \leq \sum_{i=1}^m|\lambda_i|\left|\frac{|E|}{|E_{\varepsilon_0}|}\val_{G_{\varepsilon_0},P_i}(a) - \val_{G,P_i}(a)\right| \\
& \leq \sum_{i=1}^m|\lambda_i|\varepsilon_0(d_G|\zs_a| + \vol_G(\zs_a)) \\
& \leq  \varepsilon(d_G|\zs_a| + \vol_G(\zs_a)),
\end{align*}
where the second line is due to the triangle inequality, the third is due to (\ref{eq:each_pi}) and the fourth is by the definition of $\varepsilon_0$.

Furthermore, since $m$ and all $\lambda_i$ are constants,
  \[|E_{\varepsilon_0}| =
  O\left(n\varepsilon_0^{-2}\log\frac{1}{\varepsilon_0}\right) =
  O\left(n\varepsilon^{-2}\log\frac{1}{\varepsilon}\right),\]
and so $G_{\varepsilon_0}$ is a witness for the $P$-sparsification of $G$.
\end{proof}

The core of the proof of Theorem~\ref{thm:additive_boolean} is in the next proposition, which establishes the result for Boolean predicates on even uniformity hypergraphs, with a small restriction.

\begin{proposition}
\label{prp:even_not_all_one}
Let $k$ be an even number and $G$ be a $k$-uniform hypergraph. Let $P$ be a $k$-ary Boolean predicate with $P(1,1,\hdots,1) = 0$. Then for every $\varepsilon>0$, $G$ admits $P$-sparsification with error $\varepsilon$ using $O\left(n\varepsilon^{-2}\log\frac{1}{\varepsilon}\right)$ hyperedges.
\end{proposition}

\begin{proof}
Let $\varepsilon>0$. We consider $\gamma(G)$, the $k$-partite $k$-fold cover of $G$. Let $\gamma(G)_\varepsilon$ be a subhypergraph of $\gamma(G)$ promised by Proposition~\ref{prp:directed_cut}, and $G_\varepsilon = (V,E_\varepsilon)$ the corresponding subhypergraph of $G$, i.e.\ the subhypergraph which satisfies $\gamma(G_\varepsilon) = \gamma(G)_\varepsilon$ (by taking the hyperedges corresponding to the ones of $\gamma(G)_\varepsilon$).

Let $a:V\to\{0,1\}$. For every subset $T\subseteq [k]$, we look at the
  assignment $a_T:V^\gamma\to\{0,1\}$ defined by $a_T(v^{(i)})=0$ if $i\in T$
  and  $a(v)=0$, and $a_T(v^{(i)})=1$ otherwise.
We therefore have
\begin{equation}
\label{eq:cover_sparsify}
\left|\frac{|E^\gamma|}{|E^\gamma_\varepsilon|}\val_{\gamma(G)_\varepsilon, \cut}(a_T) - \val_{\gamma(G), \cut}(a_T)\right| \leq \varepsilon(d_{\gamma(G)}|\zs_{a_T}| + \vol_{\gamma(G)}(\zs_{a_T})).
\end{equation}

Define the vector $u_T\in \mathbb{R}^{2^k}$ as follows:
  \[u_T[j] = \left\{\begin{array}{lr}
1&\qquad T\cap \zeros(j)\ne\emptyset,[k] \\ 
0&\qquad otherwise
  \end{array}\right..\]
In other words, the vector $u_T$ is 1 in index $j$ if and only if there exists an index $i\in T$ in which the binary representation of $j$ has a zero, with the exception of $u_{[k]}[0] = 0$. Denote by $P_T$ the predicate corresponding to $u_T$, that is $P_T(\bin_k(j)) = 1 \iff u_T[j] = 1$. Observe that
  \[\val_{\gamma(G),\cut}(a_T) = \val_{G,P_T}(a),\]
since they both count exactly hyperedges $(v_1,\hdots ,v_k)$ which have some
  vertex $v_i$ with $a(v_i) = 0$ with $i\in T$, but if $T = [k]$ then they do
  not count hyperedges which have $a(v_i) = 0$ for all $i=1,\hdots,k$ (see
  example in Figure~\ref{fig:bipartite_assignment}). The same is true for any
  hypergraph, and in particular for $G_\varepsilon$, that is
  \[\val_{\gamma(G_\varepsilon),\cut}(a_T) = \val_{G_\varepsilon,P_T}(a).\]
Putting these results in (\ref{eq:cover_sparsify}), we get
\begin{align*}
\left|\frac{|E|}{|E_\varepsilon|}\val_{G_\varepsilon,P_T}(a) - \val_{G,P_T}(a)\right| & \leq\varepsilon(d_{\gamma(G)}|\zs_{a_T}| + \vol_{\gamma(G)}(\zs_{a_T})) \\
& \leq\varepsilon(d_G|\zs_a| + \vol_G(\zs_a)),
\end{align*}
so $G$ admits $P_T$ sparsification with error $\varepsilon$ using
  $O\left(n\varepsilon^{-2}\log\frac{1}{\varepsilon}\right)$ hyperedges for
  every $T\subseteq [k]$, and for every $\varepsilon$ the sparsification is
  witnessed by the same subhypergraph $G_\varepsilon$. (Notice that
  Proposition~\ref{prp:directed_cut}, when applied to $\gamma(G)$ which has $kn$
  vertices, gives us a subhypergraph with
  $O\left(kn\varepsilon^{-2}\log\frac{1}{\varepsilon}\right)$ hyperedges, and
  recall that $k$ is a constant.)

\begin{figure}
\centering
\begin{tikzpicture}  
  [scale=.85,auto=center,every node/.style={draw, circle, minimum size=13mm}]  
  \node[fill=lightgray] (a1) at (1,2) {$\zs_a^{(0)}$};  
  \node (a2) at (1,0)  {$\overline{\zs}_a^{(0)}$};
  \node (a3) at (3,2)  {$\zs_a^{(1)}$};
  \node (a4) at (3,0)  {$\overline{\zs}_a^{(1)}$};
  \node (a5) at (5,2)  {$\zs_a^{(2)}$};
  \node (a6) at (5,0)  {$\overline{\zs}_a^{(2)}$};
  \node[fill=lightgray] (a7) at (7,2)  {$\zs_a^{(3)}$};
  \node (a8) at (7,0)  {$\overline{\zs}_a^{(3)}$};
  \node (a9) at (9,2)  {$\zs_a^{(4)}$};
  \node (a10) at (9,0)  {$\overline{\zs}_a^{(4)}$};
  \node[draw=none] (a11) at (11,2)  {};
  \node[draw=none] (a12) at (11,0)  {};
  \node[draw=none] (a0) at (12,1)  {\Huge{$\hdots$}};
  \node[draw=none] (a13) at (13,2)  {};
  \node[draw=none] (a14) at (13,0)  {};
  \node (a15) at (15,2)  {$\zs_a^{(k-2)}$};
  \node (a16) at (15,0)  {$\overline{\zs}_a^{(k-2)}$};
  \node[fill=lightgray] (a17) at (17.5,2)  {$\zs_a^{(k-1)}$};
  \node (a18) at (17.5,0)  {$\overline{\zs}_a^{(k-1)}$};

  \draw (a1) -- (a3);
  \draw (a1) -- (a4);
  \draw[color=red,dashed,line width=0.5mm] (a2) -- (a3);
  \draw[color=green,dotted,line width=0.5mm] (a2) -- (a4);
  \draw (a3) -- (a5);
  \draw[color=red,dashed,line width=0.5mm] (a3) -- (a6);
  \draw[color=green,dotted,line width=0.5mm] (a4) -- (a5);
  \draw (a4) -- (a6);
  \draw[color=green,dotted,line width=0.5mm] (a5) -- (a7);
  \draw (a5) -- (a8);
  \draw (a6) -- (a7);
  \draw[color=red,dashed,line width=0.5mm] (a6) -- (a8);
  \draw (a7) -- (a9);
  \draw[color=green,dotted,line width=0.5mm] (a7) -- (a10);
  \draw[color=red,dashed,line width=0.5mm] (a8) -- (a9);
  \draw (a8) -- (a10);
  \draw[color=red,dashed,line width=0.5mm] (a9) -- (a11);
  \draw (a9) -- (a12);
  \draw (a10) -- (a11);
  \draw[color=green,dotted,line width=0.5mm] (a10) -- (a12);
  \draw (a13) -- (a15);
  \draw[color=red,dashed,line width=0.5mm] (a13) -- (a16);
  \draw[color=green,dotted,line width=0.5mm] (a14) -- (a15);
  \draw (a14) -- (a16);
    \draw (a15) -- (a17);
  \draw[color=green,dotted,line width=0.5mm] (a15) -- (a18);
  \draw (a16) -- (a17);
  \draw[color=red,dashed,line width=0.5mm] (a16) -- (a18);
\end{tikzpicture}  

\caption{An example of a representation of an assignment on $\gamma(G)$.
  $\zs_a^{(i)}$ consists of all vertices in $V^{(i)}$ which are a copy of a
  vertex $v\in V$ with $a(v) = 0$, and $\overline{\zs}_a^{(i)}$ consists of the
  rest of $V^{(i)}$. Each hyperedge has a unique path from left to right (but a
  path might belong to multiple hyperedges), choosing one of
  $\zs_a^{(i)},\overline{\zs}_a^{(i)}$ for each $i$. Each such path is also in
  1-1 correspondence with a coordinate in $u_T$. In this example $T =
  \{0,3,k-1\}$ and the shaded sets represent $a_T^{-1}(0)$. By green dotted
  lines we indicated a path corresponding to a hyperedge counted in
  $\val_{\gamma(G),\cut}(a_T)$, and by red dashed lines we indicated a path
  which does not. The green dotted path corresponds to a value of $1$ in the
  coordinate of $u_T$ with binary representation $(1,1,0,0,1,\hdots,0,1)$, and
  the red dashed path to a value $0$ in the coordinate with binary
  representation $(1,0,1,1,0,,\hdots,1,1)$. Note that if $T = [k]$ then any
  hyperedge corresponding to a path only on $\zs_a^{(i)}$ is not counted.}
\label{fig:bipartite_assignment}
\end{figure}
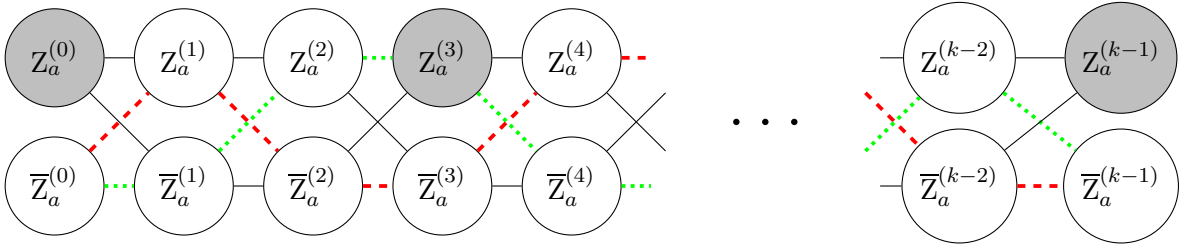

Our next goal is to show that the vector $v_P$ is a linear combination of the vectors $u_T$ for all $T\in[k]$. To show that, we show that every vector $e_r$ in the standard basis of $\mathbb{R}^{2^k}$, with $r\ne 2^k-1$, is a linear combination of these vectors. This is sufficient since the last coordinate of $v_P$ is 0 by the assumption. First we need to order the various sets $T$. We order them in the following decreasing lexicographic order $T_0,T_1,\hdots ,T_{2^k-1}$, where $T_j = \zeros(j)$, so $T_0 = [k], T_1 = [k]\setminus\{k-1\}, T_2 = [k]\setminus\{k-2\}, T_3 = [k]\setminus\{k-1,k-2\}, T_4 = [k]\setminus\{k-3\}$ and so on, until $T_{2^k-1} = \emptyset$.

Let $e_r$ be a vector in the standard basis of $\mathbb{R}^{2^k}$. We introduce the following coefficients for $0\leq m\leq 2^k-1$:
\[\lambda_{r,m} = \frac{1}{2}(-1)^{\text{Ham}(r\oplus m) + (1-\mathds{1}_{r\&m})},\]
where $\oplus,\&$ are the Xor and And binary functions respectively,\footnote{The Xor of two integers is defined as the bitwise Boolean Xor of their binary representations, where the Boolean Xor of two bits is their sum modulo 2. The And of two integers is defined the same way with the Boolean And function which is defined as And$(i,j) = 1 \iff i=j=1$.} Ham is the Hamming weight function, and $\mathds{1}_d$ returns 1 if $d\ne 0$ and 0 if $d=0$. Denote 
\[f_1(m) = \text{Ham}(r\oplus m)\quad,\quad f_2(m) = (1-\mathds{1}_{r\&m}).\]
We shall prove that 
\begin{equation}
\label{eq:lin_comb}
  e_r = \sum_{m=0}^{2^k-1}\lambda_{r,m}u_{T_m}.
\end{equation}

We start with a claim.

\noindent\textbf{Claim}: The sum of all coefficients is 0; i.e., $\sum_{m=0}^{2^k-1}\lambda_{r,m}=0$.

\emph{Proof of the claim}.\quad
  Let $b_1b_2\hdots b_k$ be the binary representation of $r$. Since $r<2^k-1$,
  there exists some $1\leq i\leq r$ for which $b_i = 0$. We can partition the
  coefficients into pairs, such that $\lambda_{r,m_1},\lambda_{r,m_2}$ is a pair if
  and only if $m_1,m_2$ differ in the $i$-th coordinate only. This is clearly a
  partition. For each pair, $f_1$ gives $m_1,m_2$ different parity values, and
  $f_2$ gives them the same value (since $b_i = 0$), so
  $\lambda_{r,m_1},\lambda_{r,m_2}$ have opposite signs, so their sum is zero. This is true for every pair, so the overall sum is zero, and the claim is proved.
\emph{(End of the proof of the claim.)}

We prove (\ref{eq:lin_comb}) coordinate-wise. First we look at the coordinate
$r$. Consider the set $W$ of all vectors $u_{T_m}$ for which the coordinate $r$
is 0. If we show that the sum of the corresponding coefficients of the vectors in $W$
is $-1$, using the claim we will deduce the result in this case. We distinguish 2 cases:

\textbf{Case (I): $r=0$.} By the definition of $u_T$, in this case the set $W$
contains two vectors, $u_{[k]}$ and $u_\emptyset$. The corresponding
coefficients are $\lambda_{r,0} = -\frac{1}{2}$ and $\lambda_{r,2^k-1} = -\frac{1}{2}$ (since $k$ is even), which sum up to $-1$.

\textbf{Case (II): $r>0$.} As in the proof of the claim, let
$b_1b_2\hdots b_k$ be the binary representation of $r$, and choose a coordinate
$1\leq i\leq k$ for which $b_i=1$. Partition the vectors in $W$ into pairs where
$u_{T_{m_1}},u_{T_{m_2}}$ is a pair if and only if $m_1,m_2$ differ in the
$i$-th coordinate only. This is clearly a partition of all vectors, and by the
definition of $u_T$, each such pair is either contained in $W$ or disjoint from
$W$ so this is indeed a partition of $W$. (Note that $u_{T_{m_1}}[r]$ is
determined by $T_{m_1}\cap \zeros(r)$ which is in fact $\zeros(m_1)\cap
\zeros(r)$, and the same for $m_2$. Since $m_1,m_2$ differ in the $i$-th
coordinate only, and $r$ is not zero in this coordinate, this coordinate
contributes nothing to the intersections, and so both these intersections are
empty or non-empty together. The intersection never equals $[k]$ since $r>0$.)
For every such pair in $W$, if it does not contain the negation of bin$(r)$,
then there is some other index $j\ne i$ in which $r,m_1,m_2$ are all
1. (This is because in all other coordinates $m_1,m_2$ are equal, and
since they are not the negation of $r$, there is some coordinate $j\ne i$ in
which they are equal to the $j$-th coordinate of $r$. These coordinates cannot
be all 0, since this would imply $u_{T_{m_1}},u_{T_{m_2}}\notin W$.) This
implies that $f_2$ gives $m_1,m_2$ the same value, and clearly $f_1$ gives them
different parity values, so $\lambda_{r,m_1}+\lambda_{r,m_2} = 0$. However, for
the pair which contains the negation of $r$ (this pair is clearly in $W$),
suppose without loss of generality the negation is $m_1$. Then $f_2$ gives
$m_1,m_2$ the values $1,0$ respectively, and $f_1$ gives $m_1$ an even value and
$m_2$ an odd value (since $k$ is even), and so $\lambda_{r,m_1} = \lambda_{r,m_2} = -\frac{1}{2}$, and the overall sum is $-1$.
This finishes the proof of~(\ref{eq:lin_comb}) in the coordinate $r$.

\medskip
Now let $r'\ne r$ be some other coordinate, and let $c_1c_2\hdots c_k$ be its
binary representation. First, if $r' = 2^k - 1$ then for all $m$ we have $u_{T_m}[r'] =
0$ by definition, so the linear combination of this coordinate is 0. So suppose
$r'<2^k-1$. As before let $W$ be the set of all vectors $u_{T_m}$ for which the
coordinate $r'$ is 0. We show that the sum of the corresponding coefficients is
zero, and again deduce the result using the claim. Now, there exists
some index $i$ for which $b_i\ne c_i$. Again we have two cases:

\textbf{Case (1): $b_i = 0, c_i = 1$.} Partition the vectors in $W$ into pairs
where $u_{T_{m_1}},u_{T_{m_2}}$ is a pair if and only if $m_1,m_2$ differ in the
$i$-th coordinate only. This is clearly a partition of all the vectors, and by
the definition of $u_T$, each such pair is either contained in $W$ or disjoint
from $W$, so this is indeed a partition of $W$. For every such pair in $W$,
$f_1$ gives $m_1,m_2$ different parity values, and $f_2$ gives them the same
value (since $b_i = 0$), so $\lambda_{r,m_1},\lambda_{r,m_2}$ have opposite signs, so their sum is zero. This is true for every pair in $W$, so the overall sum is zero.

\textbf{Case (2) $b_i = 1, c_i = 0$.} Here we consider two sub-cases:

\textbf{Case (2a): $r' = 0$.} The only vectors in $W$ in this case are $u_{[k]}$
and $u_\emptyset$. The corresponding coefficients are $\lambda_{r,0} =
\frac{1}{2}(-1)^{\text{Ham}(r) + 1}$ and $\lambda_{r,2^k-1} =
\frac{1}{2}(-1)^{\text{Ham}(\neg r)}$, where $\neg r$ denotes the negation of
the binary representation of $r$. Since $k$ is even, we know that $r,\neg r$ have the same parity, and so the sum of the two coefficients is 0.

\textbf{Case (2b): $r' \ne 0$.} Choose some $j$ for which $c_j = 1$. Partition
the vectors in $W$ into pairs where $u_{T_{m_1}},u_{T_{m_2}}$ is a pair if and
only if $m_1,m_2$ differ in the $j$-th coordinate only. The argument for this
being a partition of $W$ is similar to the argument in Case (1). For each pair in $W$, $f_1$ gives $m_1,m_2$ a different parity as always, and $f_2$ gives them the same value, since $r,m_1,m_2$ are all 1 in the index $i$ (similar argument as before), so the sum of coefficients is 0 for each pair, and so for all coefficients corresponding to vectors in $W$.

This finishes the proof of (\ref{eq:lin_comb}), and so $v_P$ is a linear combination of the vectors $u_T$. From the result above and Lemma~\ref{lem:sparse_by_linear_combination} we deduce that $G$ admits $P$-sparsification with error $\varepsilon$ using $O\left(n\varepsilon^{-2}\log\frac{1}{\varepsilon}\right)$ hyperedges, as required.
\end{proof}

To complete the picture for even $k$, we reduce to
Proposition~\ref{prp:even_not_all_one} by a simple ``complementarity trick''

\begin{proposition}
\label{prop:all_even_hypegraphs}
Let $k$ be an even number, and $G$ a $k$-uniform hypergraph. Let $P$ be a $k$-ary Boolean predicate. Then for every $\varepsilon>0$, $G$ admits $P$-sparsification with error $\varepsilon$ using $O\left(n\varepsilon^{-2}\log\frac{1}{\varepsilon}\right)$ hyperedges.
\end{proposition}

\begin{proof}
If $P(1,1,\hdots,1) = 0$ then we are done by
  Proposition~\ref{prp:even_not_all_one}. Otherwise, we have $P(1,1,\hdots,1) =
  1$, and consider $\overline{P}:\{0,1\}^k\to\{0,1\}$ defined by $\overline{P}(b_1,\hdots,b_k) = 1 - P(b_1,\hdots,b_k)$, so $v_{\overline{P}}$ is the negation of $v_P$. Since $\overline{P}$ has $\overline{P}(1,1,\hdots,1) = 0$, Proposition~\ref{prp:even_not_all_one} applies, and gives us a subhypergraph $G_\varepsilon$ for each $\varepsilon>0$, such that for every assignment $a:V\to \{0,1\}$ we have
\begin{equation}
\label{eq:sparse_negation}
\left|\frac{|E|}{|E_\varepsilon|}\val_{G_\varepsilon, \overline{P}}(a) - \val_{G, \overline{P}}(a)\right| \leq \varepsilon(d_G|\zs_a| + \vol_G(\zs_a)).
\end{equation}
Now, since $v_P + v_{\overline{P}} = \textbf{1}$ we have
  \[\val_{G,P}(a) + \val_{G,\overline{P}}(a) = \langle v_P,v_{G,a}\rangle +
  \langle v_{\overline{P}},v_{G,a}\rangle = \langle \textbf{1}, v_{G,a}\rangle =
  |E|,\]
and the same is true for $G_\varepsilon$, and so we get
\begin{equation}
\label{eq:negate_G}
\val_{G,\overline{P}}(a) = |E| - \val_{G,P}(a)
\end{equation}
and
\begin{equation}
\label{eq:negate_G_e}
\val_{G_\varepsilon,\overline{P}}(a) = |E_\varepsilon| - \val_{G_\varepsilon,P}(a).
\end{equation}
Using (\ref{eq:negate_G}) and (\ref{eq:negate_G_e}) in (\ref{eq:sparse_negation}), we get
  \[\left|\frac{|E|}{|E_\varepsilon|}(|E_\varepsilon| -
  \val_{G_\varepsilon,P}(a)) - (|E| - \val_{G,P}(a))\right|\leq
  \varepsilon(d_G|\zs_a| + \vol_G(\zs_a)),\]
and after rearranging,
  \[\left|\frac{|E|}{|E_\varepsilon|}\val_{G_\varepsilon,P}(a) -
  \val_{G,P}(a)\right|\leq \varepsilon(d_G|\zs_a| + \vol_G(\zs_a)),\]
as required.
\end{proof}

The final piece in the jigsaw
shows how to reduce sparsification of $k$-uniform
hypergraphs with $k$ odd to the case of $(k+1)$-uniform hypergraphs by adding a
universal vertex and extending the original predicate by one dimension. 

\begin{proposition}
\label{prp:all_odd_hypergraphs}
Let $k$ be an odd number, and $G = (V,E)$ a $k$-uniform hypergraph. Let $P$ be a $k$-ary Boolean predicate. Then for every $\varepsilon>0$, $G$ admits $P$-sparsification with error $\varepsilon$ using $O\left(n\varepsilon^{-2}\log\frac{1}{\varepsilon}\right)$ hyperedges.
\end{proposition}

\begin{proof}
Let $\varepsilon>0$ and denote $\varepsilon_0 = \varepsilon\frac{k}{k+1}$.
  Consider the hypergraph $G' = (V',E')$ defined by
  \[V' = V\cup \{v_0\}\quad,\quad E' = \{(v_1,\hdots,v_k,v_0):
  (v_1,\hdots,v_k)\in E\},\]
where $v_0\notin V$ is a new vertex. Clearly $G'$ is a $(k+1)$-uniform hypergraph, and $k+1$ is even. We define a new predicate $P':\{0,1\}^{k+1}\to \{0,1\}$ by
  \[P'(b_1,\hdots,b_{k+1}) = \left\{\begin{array}{lr}
1&\qquad P(b_1,\hdots,b_k) = 1, b_{k+1} = 1 \\ 
0&\qquad \text{otherwise}
  \end{array}\right..\]
We may therefore apply Proposition~\ref{prop:all_even_hypegraphs} for $G',P',\varepsilon_0$ and deduce that for every assignment $a':V'\to\{0,1\}$ we have
\begin{equation}
\label{eq:extra_vertex}
\left|\val_{G'_{\varepsilon_0}, P'}(a') - \val_{G', P'}(a')\right|\leq \varepsilon_0(d_{G'}|\zs_{a'}|+\vol_{G'}(\zs_{a'}))
\end{equation}
for some subhypergraph $G'_{\varepsilon_0} = (V',E'_{\varepsilon_0}\subseteq E')$ which does not depend on $a'$, and which satisfies $|E'_{\varepsilon_0}| = O\left(n\varepsilon_0^{-2}\log\frac{1}{\varepsilon_0}\right) = O\left(n\varepsilon^{-2}\log\frac{1}{\varepsilon}\right)$.

Let $G_{\varepsilon_0} = (V,E_{\varepsilon_0}\subseteq E)$ be the corresponding subhypergraph of $G$ (so $|E_{\varepsilon_0}| = O\left(n\varepsilon^{-2}\log\frac{1}{\varepsilon}\right)$), and for every assignment $a:V\to\{0,1\}$ define 
  \[a':V'\to \{0,1\}\quad,\quad a'(v) = \left\{\begin{array}{lr}
a(v)&\qquad v\in V \\
1&\qquad v = v_0 \\
  \end{array}\right..\]

Since the hyperedge $e = (v_1,\hdots,v_k)\in E$ is counted in $\val_{G,P}(a)$ if and only if the hyperedge $e' = (v_1,\hdots,v_r,v_0)\in E'$ is counted in $\val_{G',P'}(a')$, we have that
  \[\val_{G,P}(a) = \val_{G',P'}(a'),\]
  and the same is true for $G_{\varepsilon_0}$ and $G'_{\varepsilon_0}$. We get
\begin{align*}
\left|\val_{G_{\varepsilon_0}, P}(a) - \val_{G, P}(a)\right| & = \left|\val_{G'_{\varepsilon_0}, P'}(a') - \val_{G', P'}(a')\right| \\
& \leq \varepsilon_0(d_{G'}|\zs_{a'}|+\vol_{G'}(\zs_{a'})) \\
& = \varepsilon_0\left(\frac{(k+1)|E'|}{|V'|}|\zs_a|+\vol_G(\zs_a)\right) \\ 
& = \varepsilon_0\left(\frac{(k+1)|E|}{|V|+1}|\zs_a|+\vol_{G}(\zs_a)\right)\\
& \leq \varepsilon_0\frac{k+1}{k}\left(\frac{k|E|}{|V|}|\zs_a|+\vol_{G}(\zs_a)\right) \\
& = \varepsilon\left(d_G|\zs_a|+\vol_{G}(\zs_a)\right),
\end{align*}
where the second line is due to (\ref{eq:extra_vertex}), the next two lines are
  by the definition of $G'$ and $a'$, the fifth line is by rearranging, and the last line is by the definition of the average degree of $G$. We get the required result.
\end{proof}

Propositions~\ref{prop:all_even_hypegraphs} and~\ref{prp:all_odd_hypergraphs} complete the proof of
Theorem~\ref{thm:additive_boolean}.

\begin{remark} \label{rmk:same_subhypergraph}
In the proof of Proposition~\ref{prp:even_not_all_one} the hypergraph
$G_\varepsilon$ was chosen independently of the predicate $P$. Since
Propositions~\ref{prop:all_even_hypegraphs} and~\ref{prp:all_odd_hypergraphs} reduce to that case, we have in fact
shown that for every $\varepsilon>0$, Theorem~\ref{thm:additive_boolean} is
witnessed by the same subhypergraph $G_\varepsilon$ for all different
predicates $P$. This will be important in the proof of
  Theorem~\ref{thm:larger_domains}.
\end{remark}

\begin{remark} \label{rmk:not_constant}
We note that our main result, Theorem~\ref{thm:additive_boolean}, extends  
Theorem~\ref{thm:additive_cut} in the regime where $k$ is a constant, which is
the main focus of this paper. However, Theorem~\ref{thm:additive_cut} also
works for non-constant $k$~\cite{additive}. If $k$ is not a constant, it can
be seen from the proof of Lemma~\ref{lem:sparse_by_linear_combination} that
the number of hyperedges of the sparse subhypergraph is multiplied by a factor
of $O(m^2)$ (since $O(m)$  is the proportion between $\varepsilon$ and
$\varepsilon_0$ given that the coefficients $\lambda_i$ are constant). In
Proposition~\ref{prp:even_not_all_one} we have $m = 2^k$, and so for $k$ not
constant we get an additional factor of $4^k$. Furthermore, in
Propositions~\ref{prp:directed_cut} and~\ref{prp:even_not_all_one} we obtain 
extra factors of $k$, by considering
the $k$-partite $k$-fold cover.
While the regime with non-constant $k$ is interesting for cuts, for arbitrary
predicates one needs to be careful about representation as the natural
(explicit) representation of (non-symmetric) predicates requires exponential
  space in the arity $k$.
\end{remark}

\begin{remark}
  As observed by one of the reviewers, our sparsification result
  (Theorem~\ref{thm:additive_boolean}) actually shows sparsification under a
  stronger notion of sparsification, in which the right-hand side in~(\ref{eq:main_sparsify}) in Definition~\ref{def:additive_hypergraphs}
  is tighter. Namely, in the notation of
  Definition~\ref{def:additive_hypergraphs},
  we can require that
\begin{equation}
\label{eq:main_sparsify2}
\left|\frac{|E|}{|E_\varepsilon|}\val_{G_\varepsilon,P}(a) -
  \val_{G,P}(a)\right|\leq 
  \varepsilon\min(d_G|\zs_a| + \vol_G(\zs_a),d_G|\zs_{a'}|+\vol_G(\zs_{a'})),
\end{equation}
  where $a'(v)=1-a(v)$ for every $v\in V$. In detail, Theorem~\ref{thm:additive_cut}
  works for any assignment and thus in particular for $a'$, the value of the
  $\cut$ predicate is the same on $a$ and $a'$ (and thus also the left-hand side
  of~(\ref{eq:main_sparsify2}) is the same for $a$ and $a'$), and the rest is
  reductions that preserve~(\ref{eq:main_sparsify2}).
\end{remark}

\section{Sparsification of Non-Boolean Predicates}
\label{sec:non-Boolean}

We now focus on non-Boolean predicates; i.e., predicates of the form
$P:D^k\to\{0,1\}$ with $|D|>2$. Without loss of generality, we assume $D=[q]$
for some $q\geq 2$. The most natural way of generalising
Theorem~\ref{thm:additive_boolean} to larger domains appears to be to use the
same bound with $\zs_a = \{v\in V: a(v) = 0\}$. This, however, cannot give 
the desired sparsification result (cf. Section~\ref{sec:optimal}). Instead we use a
different and somewhat weaker kind of generalisation of the Boolean case, and
show that all hypergraphs are still sparsifiable with respect to all predicates
using this definition.

\begin{definition}
\label{def:all_but_one}
  Let $P:D^k\to\{0,1\}$ be a $k$-ary predicate where $D=[q]$. We say that a $k$-uniform
  hypergraph $G=(V,E)$ \emph{admits all-but-one $P$-sparsification} with error $\varepsilon$ using $O(f(n,\varepsilon))$ hyperedges if there exists a subhypergraph $G_\varepsilon = (V,E_\varepsilon\subseteq E)$ with $|E_\varepsilon| = O(f(n,\varepsilon))$ such that for every assignment $a:V\to D$ we have
\begin{equation}
\label{eq:large_domain_sparsification}
\left|\frac{|E|}{|E_\varepsilon|}\val_{G_\varepsilon,P}(a) - \val_{G,P}(a)\right| \leq \varepsilon(d_G|M_a|+\vol_G(N_a)),
\end{equation}
where $M_a$ is the largest set among the sets $\{v\in V: a(v) = i\}$, $N_a$ is the set with the largest volume among the sets $\{v\in V: a(v) = i\}$ for $0\leq i\leq q-2$, and $d_G$ is the average degree in $G$.
\end{definition}

Observe that the maximum in Definition~\ref{def:all_but_one} is over $i = 0,\hdots,q-2$ without $i=q-1$, hence the name
``all-but-one''. We note that there is nothing special about $q-1$ and any value
from $[q]$ could be chosen in Definition~\ref{def:all_but_one}.

Under Definition~\ref{def:all_but_one}, Theorem~\ref{thm:additive_boolean}
generalises.

\begin{theorem}
\label{thm:larger_domains}
For every $k$-uniform hypergraph $G = (V,E)$, every $k$-ary predicate
  $P:D^k\to\{0,1\}$ with $D=[q]$ ($k,q$ are constants), and every $\varepsilon>0$, $G$ admits $P$ all-but-one sparsification with error $\varepsilon$ using $O\left(n\varepsilon^{-2}\log\frac{1}{\varepsilon}\right)$ hyperedges.
\end{theorem}

Note that in the case of $q=2$ we have $P:\{0,1\}^k\to \{0,1\}$, and Definition~\ref{def:all_but_one} and Theorem~\ref{thm:larger_domains} coincide with Definition~\ref{def:additive_hypergraphs} and Theorem~\ref{thm:additive_boolean}. This is because when $q=2$ the definitions of $M_a,N_a$ coincide with the definition of $\zs_a$ in the Boolean case. 

In order to prove Theorem~\ref{thm:larger_domains}, we will generalise our notations from
Section~\ref{sec:main}. For a $k$-ary predicate $P:D^k\to
\{0,1\}$ we consider the vector $v_P\in \mathbb{R}^{q^k}$, defined by 
$v_P[i] = P(\rep_{q,k}(i))$.
For a given hypergraph $G$ and an assignment $a$, we consider the vector
$v_{G,a}\in \mathbb{R}^{q^k}$ defined by 
$v_{G,a}[i] = \left|\{(v_1,\hdots,v_k)\in E: (a(v_1),\hdots,a(v_k)) =
\rep_{q,k}(i)\}\right|$.
In other words, each coordinate of $v_{G,a}$ counts the hyperedges in $G$ whose vertices are assigned some specific set of values by $a$. Under these notations, just as before, we get
$\val_{G,P}(a) = \langle v_P, v_{G,a}\rangle$,
where $\langle\cdot,\cdot\rangle$ is the standard inner product in $\mathbb{R}^{q^k}$.

We start by proving the result for \emph{singleton predicates}, i.e.\ for predicates $P$ such that $v_P = e_r$ for some $0\leq r\leq q^k-1$.

\begin{lemma}
\label{lem:larger_domain_singletons}
Let $G = (V,E)$ be a $k$-uniform hypergraph, and $P:D^k\to\{0,1\}$ a $k$-ary
  singleton predicate with $D=[q]$ ($k,q$ are constants). For every $\varepsilon>0$, $G$ admits $P$ all-but-one sparsification with error $\varepsilon$ using $O\left(n\varepsilon^{-2}\log\frac{1}{\varepsilon}\right)$ hyperedges.
\end{lemma}

\begin{proof}
Denote $\varepsilon_0 = \frac{\varepsilon}{q}$. Let $\gamma(G) = (V^\gamma,E^\gamma)$ be the $k$-partite $k$-fold cover of $G$, and let $\gamma(G)_{\varepsilon_0} = (V^\gamma,E_{\varepsilon_0}^\gamma\subseteq E^\gamma)$ be the subhypergraph promised by Theorem~\ref{thm:additive_boolean}. From Remark~\ref{rmk:same_subhypergraph} we know that this is the same subhypergraph for all predicates, and it does not depend on the choice of $P$. As before, let $G_{\varepsilon_0} = (V,E_{\varepsilon_0}\subseteq E)$ be the subhypergraph of $G$ which satisfies $\gamma(G_{\varepsilon_0}) = \gamma(G)_{\varepsilon_0}$.

Let $r$ be the integer for which $v_P = e_r$, and denote $u_r = \rep_{q,k}(r)$.
  Consider the set $T = \{i\in [k]: u_r[i] = q-1\}$. For each assignment $a:V\to [q]$, we want to find a Boolean assignment $a_r:V^\gamma\to \{0,1\}$ which does not assign 0 to any vertices $v^{(i)}$ which have $a(v) = q-1$, but which can also be used to show $P$-sparsification. We define
  \[a_r(v^{(i)}) = \left\{\begin{array}{lr}
0&\qquad i\notin T\text{ and }a(v) = u_r[i]\\ 
0&\qquad i \in T\text{ and } a(v) \ne q-1\\
1&\qquad \text{otherwise}
  \end{array}\right..\]
We also define a $k$-ary Boolean predicate $P_r:\{0,1\}^k\to\{0,1\}$ to only have a single truth value (a singleton) which is $(b_1,\hdots,b_k)$ where 
$b_i=0$ if $i\not\in T$ and $b_i=1$ otherwise.

Observe that 
\begin{equation}
\label{eq:larger_val_g}
\val_{G,P}(a) = \val_{\gamma(G),P_r}(a_r),
\end{equation}
since both count the same hyperedges. The same is true for any hypergraph, and specifically for $G_{\varepsilon_0}$, that is
\begin{equation}
\label{eq:larger_val_gamma_g}
\val_{G_{\varepsilon_0},P}(a) = \val_{\gamma(G_{\varepsilon_0}),P_r}(a_r).
\end{equation}
Using Theorem~\ref{thm:additive_boolean} for $\gamma(G), P_r, \varepsilon_0$ and $a_r$ we get
\begin{align*}
\left|\frac{|E|}{|E_{\varepsilon_0}|}\val_{G_{\varepsilon_0},P}(a) - \val_{G,P}(a)\right| & = \left|\frac{|E^\gamma|}{|E_{\varepsilon_0}^\gamma|}\val_{\gamma(G_{\varepsilon_0}),P_r}(a_r) - \val_{\gamma(G),P_r}(a_r)\right| \\ 
& \leq \varepsilon_0(d_{\gamma(G)}|\zs_{a_r}|+\vol_{\gamma(G)}(\zs_{a_r})) \\
& \leq \varepsilon_0\left(\frac{d_G}{k}\cdot kq|M_a|+q\cdot\vol_G(N_a)\right) \\
& = \varepsilon \left(d_G |M_a|+\vol_G(N_a)\right),
\end{align*}
where the first line follows from (\ref{eq:larger_val_g}), (\ref{eq:larger_val_gamma_g}) and Definition~\ref{def:cover}, the second line is the application of Theorem~\ref{thm:additive_boolean}, the third is again Definition~\ref{def:cover} and the definitions of $a_r,M_a,N_a$, and the last is the definition of $\varepsilon_0$. This is true for every assignment $a$. In addition we have
  \[|E_{\varepsilon_0}| = O\left(n\varepsilon_0^{-2}\log\frac{1}{\varepsilon_0}\right) = O\left(n\varepsilon^{-2}\log\frac{1}\varepsilon\right),\] so $G_{\varepsilon_0}$ is the required witness.
\end{proof}

The proof of Theorem~\ref{thm:larger_domains} is now an application of Lemma~\ref{lem:larger_domain_singletons} similar to the proof of Lemma~\ref{lem:sparse_by_linear_combination}.

\begin{proof}[Proof of Theorem~\ref{thm:larger_domains}]
The vector $v_P$ satisfies
  \[v_P = \sum_{r=0}^{q^k-1}\lambda_re_r,\] for some $\lambda_r\in\{0,1\}$ and $e_r$ vectors of the standard basis of $\mathbb{R}^{q^k}$. Therefore, for every assignment $a:V\to D$, we have
\begin{equation}
\label{eq:val_g}
\val_{G,P}(a) = \langle v_P, v_{G,a}\rangle = \sum_{r=0}^{q^k-1}\lambda_r\langle e_r, v_{G,a}\rangle = \sum_{r=0}^{q^k-1}\lambda_r\val_{G,P_r}(a)
\end{equation}
where $P_r$ is the predicate corresponding to the vector $e_r$. Let $G_{\varepsilon_0} = (V,E_{\varepsilon_0}\subseteq E)$ be the subhypergraph of $G$ promised by Lemma~\ref{lem:larger_domain_singletons} for $\varepsilon_0 = \frac{\varepsilon}{q^k}$. Note that this is the same subhypergraph for all predicates $P_r$ (see proof of the lemma). Equation (\ref{eq:val_g}) is true for any other hypergraph as well, and in particular $G_{\varepsilon_0}$. Using Lemma~\ref{lem:larger_domain_singletons} for each $P_r$, we get
\begin{align*}
\left|\frac{|E|}{|E_{\varepsilon_0}|}\val_{G_{\varepsilon_0},P}(a) - \val_{G,P}(a)\right| & = \left|\frac{|E|}{|E_{\varepsilon_0}|}\sum_{r=0}^{q^k-1}\lambda_r\val_{G_{\varepsilon_0},P_r}(a) - \sum_{r=0}^{q^k-1}\lambda_r \val_{G,P_r}(a)\right| \\ 
& \leq \sum_{r=0}^{q^k-1}\lambda_r \left|\frac{|E|}{|E_{\varepsilon_0}|}\val_{G_{\varepsilon_0},P_r}(a) - \val_{G,P_r}(a)\right| \\
& \leq \sum_{r=0}^{q^k-1}\lambda_r\varepsilon_0(d_G|M_a|+\vol_G(N_a)) \\
& \leq q^k\varepsilon_0(d_G|M_a|+\vol_G(N_a)) \\
& = \varepsilon(d_G|M_a|+\vol_G(N_a)),
\end{align*}
where the first line follows from (\ref{eq:val_g}) for the different hypergraphs, the second line from the triangle inequality, the third from Lemma~\ref{lem:larger_domain_singletons}, the fourth is due to $\lambda_r\in\{0,1\}$ for all $r$, and the last is the definition of $\varepsilon_0$. Again,
  \[|E_{\varepsilon_0}| = O\left(n\varepsilon_0^{-2}\log\frac{1}{\varepsilon_0}\right) = O\left(n\varepsilon^{-2}\log\frac{1}\varepsilon\right),\] 
  so we have found an appropriate subhypergraph of $G$.
\end{proof}

\begin{remark}
  Similarly to Remark~\ref{rmk:not_constant}, if $k$ and $q$ are not constant we
  get an additional factor of $q^{2k}$.
\end{remark}

\section{Optimality of All-But-One Sparsification}
\label{sec:optimal}

One might wonder if there is a different, perhaps stronger way to define
sparsification for predicates on non-Boolean domains. The following example shows
that all-but-one sparsification is optimal.

For a hypergraph $G = (V,E)$ and a fixed assignment $a:V\to[q]$ denote $S_i =
\{v\in V: a(v) = i\}$ (so $S_0 = \zs_a$). The definition of all-but-one
sparsification lets us take a bound which depends on the sizes and volumes of
all the sets $S_i$ except for $S_{q-1}$. In fact, if we try to take a bound
which depends on fewer of these sets, the definition fails to generalise even
the most basic case of the $\cut$ predicate. To see this, it is sufficient to
consider the graph case, i.e.\ $k=2$.
Let us suppose, without loss of generality, that our bound does not depend on
$S_{q-2},S_{q-1}$. Consider the predicate $\cut:[q]^2\to\{0,1\}$ defined by
$\cut(x,y)=1\iff x\ne y$. A simple (but lengthy) argument below shows that
cliques do not have a $\cut$-sparsifier using such a definition. Therefore, no
definition with a bound which depends on ``less'' is possible, under the current
assumptions.

Let $G=K_n$ be the complete graph with $n$ vertices $v_1,\ldots,v_n$. Moreover,
let $G_\varepsilon$ be a subgraph of $G$, and consider the predicate
$\cut:[q]^2\to\{0,1\}$ defined by $\cut(x,y)=1\iff x\ne y$. 

We consider two cases:

\textbf{Case (1):} All vertices in $G_\varepsilon$ have the same degree $d>0$. 

We look at the assignment $a:V\to \{0,1\}$ defined as $a(v_1)=a(v_2) = q-2$ and for all $i>2$, $a(v_i) = q-1$. If $G_\varepsilon$ is a $\cut$-sparsifier of $G$ then
\[\left|\frac{|E|}{|E_\varepsilon|}\val_{G_\varepsilon,\cut}(a) -
\val_{G,\cut}(a)\right| \leq \varepsilon(d_G|S| + \vol_G(S)) = 0,\] where $S$ is some set which depends only on the empty sets $S_0,\hdots,S_{q-3}$. If $v_1,v_2$ are neighbours in $G_\varepsilon$ this implies
\[\frac{n-1}{d} = \frac{\frac{n(n-1)}{2}}{\frac{nd}{2}} = \frac{|E|}{|E_\varepsilon|} = \frac{\val_{G,\cut}(a)}{\val_{G_\varepsilon,\cut}(a)} = \frac{2(n-2)}{2(d-1)} = \frac{n-2}{d-1},\]  and if they are not then
\[\frac{n-1}{d} = \frac{|E|}{|E_\varepsilon|} =
\frac{\val_{G,\cut}(a)}{\val_{G_\varepsilon,\cut}(a)} =  \frac{2(n-2)}{2d} =
\frac{n-2}{d}.\] The second option is a contradiction, and the first option implies $d=n-1$, which means $|E_\varepsilon| = |E|$, so $G_\varepsilon$ is not a sparsifier.

\textbf{Case (2):} There exist two vertices $v_i,v_j$ with degrees $d_i\ne d_j$ in $G_\varepsilon$. 

We look at two assignments
$a_1:V\to\{0,1\}$ defined by $a_1(v_i)=q-2$ and $a_1(v)=q-1$ for $v\ne v_i$,
and
$a_2:V\to\{0,1\}$ defined by $a_2(v_j)=q-2$ and $a_2(v)=q-1$ for $v\ne v_j$.
Since 
\[\frac{\val_{G,\cut}(a_1)}{\val_{G_\varepsilon,\cut}(a_1)} = \frac{n-1}{d_i}\ne
\frac{n-1}{d_j} = \frac{\val_{G,\cut}(a_2)}{\val_{G_\varepsilon,\cut}(a_2)},\]
at least one side of the inequality is different from $\frac{|E|}{|E_\varepsilon|}$. Suppose without loss of generality this is the left hand side. Then
\[\left|\frac{|E|}{|E_\varepsilon|}Val_{G_\varepsilon,\cut}(a_1) -
Val_{G,\cut}(a_1)\right| > 0 = \varepsilon(d_G|S| + \vol_G(S)),\]
where again $S$ is some set not depending on $S_{q-2},S_{q-1}$, and so $G_\varepsilon$ is not a $\cut$-sparsifier of $G$.

Note that the same argument works for any predicate $P$ with $P(q-2,q-1) =
P(q-1,q-2) = 1$ and $P(q-2,q-2) = P(q-1,q-1) = 0$. Thus if a definition does not
depend on more than just $S_{q-2},S_{q-1}$, it specifically does not depend on
these two, so the same argument still works. 

\section*{Acknowledgements}

We would like to thank the anonymous referees of both the
conference~\cite{pz21:esa} and this full version of the paper.

{\small
\bibliographystyle{plainurl}
\bibliography{pz23}
}

\end{document}